\documentclass[12pt]{article}
\usepackage{array,xspace,multirow,hhline,tikz,colortbl,tabularx,booktabs,fixltx2e,amsmath,amssymb,amsfonts,amsthm}
\usetikzlibrary{arrows}
\usepackage{algorithm,algorithmic}
\usepackage{eqparbox}
\usepackage{verbatim,ifthen}
\usepackage{enumitem}
\usepackage{pifont}
\usepackage{setspace}

\usepackage{natbib}
\usepackage{ifthen}
\usepackage{calrsfs,mathrsfs}
\usepackage{bbding,pifont}
\usepackage{pgflibraryshapes}
\usepackage{url}
\usepackage{hyperref}
\usepackage{geometry}
\usepackage{abstract}
\usepackage{gensymb} 
 
	\usepackage{varioref} 
		\usepackage{subfigure}

\definecolor{light-gray}{gray}{0.9}

\newtheorem{definition}{Definition}

	\newtheorem{lemma}{Lemma}%
	\newtheorem{theorem}{Theorem}%
\hypersetup{
	colorlinks=true,
	linkcolor=blue,
	filecolor=blue,      
	urlcolor=blue,
	citecolor=blue
}

\newcommand{\harisnew}[1]{\textcolor{black}{#1}}

	\newcommand\eat[1]{}

	\usepackage{enumitem}
	\setenumerate[1]{label=\rm(\it{\roman{*}}\rm),ref=({\it\roman{*}}),leftmargin=*}
	\newlength{\wordlength}

	\newcommand{\eqclass}[2][]{\ifthenelse{\equal{#1}{}}{[#2]}{[#2]_{\sim_{#1}}}}

	\newcommand{\ceil}[1]{\lceil #1 \rceil }
	\newcommand{\floor}[1]{\lfloor #1 \rfloor }

		\newcommand{\spref}{\succ\xspace}

\usepackage{enumitem}
\setenumerate[1]{label=\rm(\it{\roman{*}}\rm),ref=({\it\roman{*}}),leftmargin=*}

\newcommand{\nbh}[1][]{
	\ifthenelse{\equal{#1}{}}{\nu}{\nu(#1)}
}

\newcommand{\cstr}[1][]{
	\ifthenelse{\equal{#1}{}}{\mathscr S}{\cstr(#1)}
}

\newcommand{\choice}[1][]{
	\ifthenelse{\equal{#1}{}}{\mathit{C}}{\choice(#1)}

		\newcommand{\ml}[1][]{\ensuremath{\ifthenelse{\equal{#1}{}}{\mathit{ML}}{\mathit{ML}(#1)}}\xspace}
		\newcommand{\sml}[1][]{\ensuremath{\ifthenelse{\equal{#1}{}}{\mathit{SML}}{\mathit{SML}(#1)}}\xspace}
		\newcommand{\sd}[1][]{\ensuremath{\ifthenelse{\equal{#1}{}}{\mathit{SD}}{\mathit{SD}(#1)}}\xspace}
		\newcommand{\rsd}[1][]{\ensuremath{\ifthenelse{\equal{#1}{}}{\mathit{RSD}}{\mathit{RSD}(#1)}}\xspace}
		\newcommand{\rd}[1][]{\ensuremath{\ifthenelse{\equal{#1}{}}{\mathit{RD}}{\mathit{RD}(#1)}}\xspace}
		\newcommand{\st}[1][]{\ensuremath{\ifthenelse{\equal{#1}{}}{\mathit{ST}}{\mathit{ST}(#1)}}\xspace}
		\newcommand{\bd}[1][]{\ensuremath{\ifthenelse{\equal{#1}{}}{\mathit{BD}}{\mathit{BD}(#1)}}\xspace}
		\newcommand{\pc}[1][]{\ensuremath{\ifthenelse{\equal{#1}{}}{\mathit{PC}}{\mathit{PC}(#1)}}\xspace}
		\newcommand{\dl}[1][]{\ensuremath{\ifthenelse{\equal{#1}{}}{\mathit{DL}}{\mathit{DL}(#1)}}\xspace}
		\newcommand{\ul}[1][]{\ensuremath{\ifthenelse{\equal{#1}{}}{\mathit{UL}}{\mathit{UL}(#1)}}\xspace}

			\newcommand{\indiff}{\ensuremath{\sim}}}

			\usepackage{palatino}

\begin{document}

%
%
%
%
%
%
%
%


	\title{A characterization of \\proportionally representative committees}




	\author{Haris Aziz \ \  and \ \  Barton E. Lee\thanks{UNSW Sydney and Data61 CSIRO, Australia. Email: \href{mailto: haris.aziz@unsw.edu.au}{haris.aziz@unsw.edu.au}; \href{mailto: barton.e.lee@gmail.com}{barton.e.lee@gmail.com}. We thank Bill Zwicker for helpful comments and discussions.} }
%



	\maketitle

		\begin{abstract}
			A well-known axiom for proportional representation is Proportionality of Solid Coalitions (PSC). We characterize committees satisfying PSC as possible outcomes of the Minimal Demand rule, which generalizes an approach pioneered by Michael Dummett.
	\end{abstract}
	
	\bigskip
	
		




		\section{Introduction} 

		In multiwinner elections, a central concern is proportional representation of voters.\footnote{See, for example, the prominent electoral reform movements   by  FairVote (\url{https://www.fairvote.org})  and the Electoral Reform Society (\url{https://www.electoral-reform.org.uk}).  }
		When voters elicit ranked preferences over candidates, one particular axiom for proportional representation is Proportionality of Solid Coalitions (PSC). This axiom was advocated by \cite{Dumm84a} and has been referred to as the most important requirement for proportional representation~\citep{Tide95a,TiRi00a,Wood94a,Wood97a}.\footnote{For broader discussion on axioms for multiwinner voting, see the chapter by \citet{FSST17a}.} 

		PSC is the subject of many theoretical and empirical studies. Theoretical studies have focused on designing voting rules that satisfy PSC; these include
		single transferable vote (STV)~\citep{Tide95a}, Quota Borda System (QBS)~\citep{Dumm84a}, Schulz-STV~\citep{Schu11a}, and the Expanding Approvals Rule (EAR)~\citep{AzLe19a}.\footnote{{There is also a growing literature that explores issues of proportionality in the context of \emph{participatory budgeting}~\citep[see, e.g.,][]{AzLe21,ALT18a,FPPV21}.}} STV is the most prominent of these rules\footnote{STV is used for elections in   Australia, Ireland, India, and Pakistan.} and has attracted significant attention in the literature from both a theoretical~\citep{Gell05,Howa90,Mill07,PePe17,Ray86,VanD93} and empirical perspective~\citep{EnTo14,FMM96,LaMc05}. However, few studies consider the structure imposed by the PSC axiom on election outcomes.

		In this note, we present a characterization of PSC committees as the range of outcomes from a certain class of procedures, which we formalize and call 
 \emph{Minimal Demand (MD)} rules. {MD generalizes an approach pioneered by Michael Dummett who proposed one particular rule within our wider class. }
We also present an alternative way of viewing MD in terms of a ``Dummett Tree,'' which represents a decision tree with decisions taken at each branching of a node. Our main result is the following:  \emph{A committee satisfies PSC if and only it is a possible outcome of MD if and only if it is an outcome of a branch of a Dummett tree. }  

		This contribution is important because it provides an intuitive and tractable method for researchers to analyze the demands of PSC on committee outcomes. We hope that this characterization will be { useful as a stepping stone toward understanding the interaction between PSC and other axioms. }


		\section{Preliminaries}

		%


		We consider the standard social choice setting with a set of voters $N=\{1,\ldots, n\}$, a set of candidates $C=\{c_1,\ldots, c_m\}$ and a preference profile $\spref=(\spref_1,\ldots,\spref_n)$ such that each $\spref_i$ is linear order over $C$. {If $c_s \spref_i c_t$ then we say that voter $i$ \emph{prefers} candidate $c_s$ to candidate $c_t$.} 
		 Given $j\in \{1, \ldots, m\}$,  the \emph{$j$-prefix} of a voter's preference  is their (unordered) set of $j$ most preferred candidates.   
	The goal is to select a committee $W\subset C$ of pre-determined size $k$. 
		
		{The focus of this paper is on committees that satisfy the Proportionality of Solid Coalitions (PSC) axiom. Before formally defining PSC, we introduce the notion of a solid coalition which is central to the PSC axiom.  Intuitively, a set of voters $N'$ forms a solid coalition for a set of candidates $C'$ if every voter in $N'$ prefers every candidate in $C'$ to any candidate outside of $C'$.  Importantly, voters that form a solid coalition for a candidate-set $C'$ are not required to have identical preference orderings over candidates within $C'$ nor $C\backslash C'$.}

		\begin{definition}[Solid coalition]
		A set of voters $N'$ is a \emph{solid coalition} for a set of candidates $C'$ if for all $i\in N'$ and for any $c'\in C'$
		$$\forall c\in C\backslash C' \quad c'\spref_i c. $$
		The candidates in $C'$ are said to be \emph{{ (solidly)} supported} by the voter set $N'$, and, conversely, the voter set $N'$ is said to be \emph{{ (solidly)}  supporting} the candidate set $C'$.
		\end{definition}

{We now state the PSC definition. Informally, PSC requires that a committee $W$ ``adequately" represents the preferences of solid coalitions that are ``sufficiently large" by including some candidates that they support in relation to the size of the solid coalition.  Here ``adequately" and ``sufficiently large" are defined according to a parameter $q\in (n/(k+1), n/k]$ and leads to a hierarchy\footnote{If an outcome $W$ satisfies  $q$-PSC, then $W$ satisfies $q'$-PSC for all $q'>q$ \citep[Lemma 2]{AzLe19a}.} of PSC definitions denoted by $q$-PSC.}\footnote{See~\cite[Footnote 10]{AzLe19a} for justification of the bounds on  $q$.}

		\begin{definition}[$q$-PSC]
		Let $q\in (n/(k+1), n/k]$.  A committee $W$ satisfies $q$-PSC if for every positive integer $\ell$, and for every solid coalition $N'$ supporting a candidate subset $C'$ with size $|N'|\ge \ell q$,  the following holds
		$$|W\cap C'| \ge \min\{\ell, |C'|\}.$$
		{Given a subset $\tilde{W}\subseteq C$ and a solid coalition $N'$ supporting $C'$, we say that $N'$ has an \emph{unmet PSC demand} if $|\tilde{W}\cap C'| < \min\{ \floor{|N'|/q} ,|C'|\}$.}
		\end{definition} 
		
		{The $q$-PSC axiom captures intuitive features of proportional representation. The axiom ensures representation of minority voters so long as they share similar preferences over candidates, i.e., they form a solid coalition, and the amount of representation given to  a group of voters that form a solid coalition is (approximately) in proportion to their size.} { For the remainder of the paper, we will explore the $q$-PSC axiom for fixed $q$; hence, abusing notation slightly, we refer to $q$-PSC as simply PSC.}

			 \section{Minimal Demand Rule and the Dummett Tree}
	 
			 Dummett proposed the Quota Borda System (QBS) rule as follows.
			It examines the prefixes (of increasing sizes) of the preference lists of voters and checks if there exists a corresponding solid coalition for a set of voters. If there is such a solid set of voters, then an appropriate number of candidates with the highest Borda count are selected so as to satisfy the corresponding PSC demand.

		 We view Dummett's approach as a special case of a more general class of rules, which we call    \textbf{Minimal Demand (MD)} rules.	{ We formalize  the MD rule below in 3 steps.}  Just like in QBS,   MD rules  examine  prefixes (of increasing sizes) of the preference lists of voters and checks if there exists a corresponding solid coalition for a set of voters.\footnote{{ Note that each voter-partition in Step~1 is a solid coalition supporting their $j$-prefix.}} If there is such a solid set of voters, then a minimal subset of candidates is sequentially selected to satisfy the corresponding PSC demand.
		 

		
		For a given $q\in (n/(k+1), n/k]$, the MD rule is implemented as follows. Initialize $W=\emptyset$ and $j=1$.    
		\begin{description}[leftmargin = 1.43cm]
		\item[Step 1.] Partition the set of voters into equivalence classes where each class has the same $j$-prefix
		\item[Step 2.] If there exists an equivalence class of voters $N'\subseteq N$ with $j$-prefix  $C'$ such that $|W\cap C'|<\min\{\floor{|N'|/q}, |C'|\}$, then for any such $c'\in C'\backslash W$ update $W$ to $W\cup \{c'\}$. Repeat with the updated $W$ until no additional candidate can be added.
		\item[Step 3.] If $j<m$, update $j$ to $j+1$ and repeat from step (i). Otherwise, terminate and output $W$. 
		\end{description}
		
	  
	  Lemma~\ref{lemma: 1} verifies that the MD rule terminates and outputs a committee of size $k$.

		\begin{lemma}\label{lemma: 1}
The MD rule always terminates and outputs $W \ : \ |W|=k$. 
\end{lemma}

\begin{proof}
First, the MD rule    terminates because { the ``if'' condition in Step~2 can always be satisfied by adding all candidates in $C'$.}

Second, the MD rule outputs $W$ such that $|W|\ge k$. {This is because, at stage $j=m$, all of the voters have the same $j$-prefix, equal to $C$; hence,  the if condition in Step~2 is satisfied  for any $W \ : \ |W|<k$.}

{It remains to prove that the MD rule's output $W$ never exceeds size $k$. For sake of contradiction, suppose that $|W|>k$. Given $j\in \{1, \ldots, m\}$, let $W_j\subseteq W$ be the set of candidates elected at the termination of stage $j$.  Consider   stage $j^* \ : \  W_{j^*-1}\le k$ and $W_{j^*}>k$. At some point during stage $j^*$,  $k$ candidates were   elected by the MD rule (denote this set $W_{j^*}'$) and, yet, there exists an equivalence class of voters $N'$ with $j$-prefix $C'$ such that $|W_{j^*}'\cap C'| <\min\{\floor{|N'|/q}, |C'|\}$. By   definition of the MD rule and  PSC,  a one-to-one mapping $\phi$ can be constructed between each  of the $k$ elected candidate at this point, say $c'\in W_{j^*}'$, and a distinct set of $\ceil{q}$ voters that  all have $c'$ in their $j^*$-prefix. Now consider the equivalence class of voters $N'$.   By the  pigeon-hole principle and existence of $\phi$,  at least $|N'|/q$ candidates in $W_{j^*}'$ are contained in the $j^*$-prefix of voters in $N'$. Thus, $|W_{j^*}'\cap C'| \ge  |N'|/q $---a contradiction.}
 \end{proof}

		Since MD specifies a candidate subset from which a candidate should be selected, rather than specifying a single candidate, there is a great deal of flexibility in how these ``ties" are resolved, i.e., which candidate from the subset is selected. By considering different tie-breaking decisions, we can  attain different outcomes of MD. These outcomes can be represented in the form of a decision tree, which we will refer to as the \textbf{Dummett Tree}. 

				The Dummett tree can be viewed as a tree corresponding to how the \harisnew{(possibly) non-deterministic} MD rule can be run depending on the selection of candidates in each stage. 
		Each node along a path of the Dummett tree corresponds to a stage. 
		The Dummett tree has depth  $m$. 
		If $k$ candidates have already been selected by stage $j$, we still go over all the stages. 
		At each stage $j$ of the tree the rule only considers the $j$-prefixes of voter preferences.
		Accordingly, only those PSC demands that pertain to the most preferred $j$ candidates of each voter are considered. Each path of the tree reflects the selection of candidates. Each node in the path represents the selection of candidates at that point keeping in view the candidates already selected at nodes higher up in the path.\footnote{Equivalently, one may also include ``null" nodes, which represent stages where no candidate was selected, in the tree.}





 		\section{MD, Dummett Trees and PSC}

		We present some connections between MD, Dummett Trees and PSC.	Lemma~\ref{lemma: 2} states that each possible outcome of MD satisfies PSC. To the best of our knowledge, it is the first formal argument that MD returns a PSC committee and the outcome being PSC does not depend on the specific choices made during the algorithm.

		\begin{lemma}\label{lemma: 2}
			Each path along the Dummett Tree gives rise to  a committee satisfying PSC. Equivalently, each possible outcome of the MD rule satisfies PSC.
			\end{lemma}

				\begin{proof}
					{ Let $W$ be an outcome of the MD rule (by Lemma~\ref{lemma: 1}, $|W|=k$). For  sake of a contradiction, suppose   $W$ does not satisfy PSC:  there exists a positive integer $\ell$ and set of voters $N'$ with $|N'|\ge \ell q$ solidly supporting a candidate subset $C'$ such that  $|W\cap C'|<\min\{\ell, |C'|\}$. But then, at stage $j=|C'|$ of the MD rule, the voters in $N'$ form an equivalence class { (with $j$-prefix $C'$) that satisfies the if condition in Step~2.} Thus, the algorithm cannot proceed to the next stage while guaranteeing the output $W$---a contradiction. 	}
%
\end{proof}				

		Lemma~\ref{lemma: 2} implies the following. Any voting rule that proceeds by sequentially electing candidates  is guaranteed to satisfy PSC so long as: (1) candidates are only added if they resolve a PSC violation, and (2) candidates that resolve violation of PSC for smaller prefixes of voters preferences are elected before candidates that resolve violations for larger prefixes. 
	
			Lemma~\ref{lemma: 3} provides a converse  to Lemma~\ref{lemma: 2}: every PSC committee is a possible outcome of MD. 
	
			\begin{lemma}\label{lemma: 3}
		If  $W$ is a committee satisfying PSC, then there exists a path of the Dummett tree that selects $W$.				\end{lemma}
				\begin{proof}
					Suppose $W$ satisfies PSC.
					We simulate an MD outcome that makes decisions along the tree and selects candidates from $W$. The proof is by induction on the number of stages. At each stage, we can select candidates only from $W$ to fulfill the PSC demands of voters. { First, note that since $W$ satisfies PSC, it does not satisfy the if condition in Step~2 for any $j$.} { Now, if we have selected $W_j\subseteq W$ by the $j$-th stage and the if condition in Step~2 is satisfied, then there exists $c\in W\backslash W_j$ that can be added at Step~2.  Repeating this process leads to the outcome $W$ and, hence, there   exists a path in the Dummett tree that selects $W$. }
%
					\end{proof}

					Combining the two lemmas gives us the following  equivalence theorem. 
			
					\begin{theorem}\label{theorem: equivalence of PSC and MD and Dummett tree}
						The following three statements are equivalent.
						\begin{description}
							\item[(i)] A committee satisfies PSC.
							\item[(ii)] A committee is an outcome of some path of the Dummett tree.
							\item[(iii)] A committee is a possible outcome of MD.
			 
						\end{description}
						\end{theorem}

			{\small
	\setstretch{.75}
\bibliographystyle{aer2}
		  \bibliography{../adtbib/abb.bib,../adtbib/adt.bib}
}

	\end{document}